\documentclass[submission,copyright,creativecommons]{eptcs}
\providecommand{\event}{MFPS XXXVII}
\usepackage[T1]{fontenc}
\usepackage[utf8]{inputenc}
\usepackage[english]{babel}
\usepackage{amsmath,amsfonts,amssymb,amsthm}
\usepackage{stmaryrd}
\usepackage{cmll}
\usepackage{tikz-cd}
\usetikzlibrary{patterns,decorations.pathreplacing}
\usepackage{graphicx}
\usepackage{mathpartir}
\usepackage[shortlabels]{enumitem}
\usepackage{mathpartir}
\usepackage{cleveref}
\usepackage{macros}
\usetikzlibrary{calc,decorations.pathmorphing,shapes,arrows}
\usepackage{textgreek}

\newtheorem{theorem}{Theorem}
\newtheorem{proposition}[theorem]{Proposition}
\newtheorem{lemma}[theorem]{Lemma}
\theoremstyle{definition}
\newtheorem{definition}[theorem]{Definition}
\theoremstyle{remark}
\newtheorem{remark}[theorem]{Remark}

\renewcommand{\epsilon}{\varepsilon}

\crefname{thm}{theorem}{theorems}
\crefname{lem}{lemma}{lemmas}

\usepackage{newunicodechar}
\newunicodechar{→}{\ensuremath\to}
\newunicodechar{×}{\ensuremath\times}
\newunicodechar{¬}{\ensuremath\lnot}
\newunicodechar{ℓ}{\ensuremath\ell}
\newunicodechar{≃}{\ensuremath\simeq}
\newunicodechar{≡}{\ensuremath\equiv}
\newunicodechar{↝}{\ensuremath\leadsto}
\newunicodechar{∙}{\ensuremath\cdot}
\newunicodechar{⊗}{\ensuremath\otimes}
\newunicodechar{⊕}{\ensuremath\oplus}
\newunicodechar{⊔}{\ensuremath\sqcup}
\newunicodechar{∀}{\ensuremath\forall}
\newunicodechar{⊤}{\ensuremath\top}
\newunicodechar{⊥}{\ensuremath\bot}
\newunicodechar{∥}{\ensuremath\|}
\newunicodechar{⟦}{\ensuremath\llbracket}
\newunicodechar{⟧}{\ensuremath\rrbracket}
\newunicodechar{𝔹}{\ensuremath{\mathbb{B}}}
\newunicodechar{ℕ}{\ensuremath{\mathbb{N}}}
\newunicodechar{₀}{\ensuremath{\null_0}}
\newunicodechar{₁}{\ensuremath{\null_1}}
\newunicodechar{₂}{\ensuremath{\null_2}}
\newunicodechar{ᵢ}{\ensuremath{\null_\text{i}}}
\newunicodechar{∧}{\ensuremath\land}
\newunicodechar{∨}{\ensuremath\lor}
\newunicodechar{∘}{\ensuremath\circ}

\renewcommand\textphisymbol\textphi

\hypersetup{hidelinks}

\title{A Cartesian Bicategory of Polynomial Functors in Homotopy Type Theory}
\author{
  Eric Finster
  \institute{University of Cambridge}
  \and
  Samuel Mimram
  \institute{École polytechnique}
  \and
  Maxime Lucas
  \institute{Université Sorbonne Paris Nord}
  \and
  Thomas Seiller
  \institute{CNRS\thanks{Supported by the Île-de-France region through the DIM RFSI project CoHOp.}}
}
\def\titlerunning{A Cartesian Bicategory of Polynomial Functors in Homotopy Type Theory}
\def\authorrunning{E. Finster, S. Mimram, M. Lucas, T. Seiller}

\begin{document}
\maketitle
\begin{abstract}
  Polynomial functors are a categorical generalization of the usual notion of
polynomial, which has found many applications in higher categories and type
theory: those are generated by polynomials consisting a set of monomials built
from sets of variables. They can be organized into a cartesian bicategory, which
unfortunately fails to be closed for essentially two reasons, which we address
here by suitably modifying the model. Firstly, a naive closure is too large to
be well-defined, which can be overcome by restricting to polynomials which are
finitary. Secondly, the resulting putative closure fails to properly take the
2-categorical structure in account. We advocate here that this can be addressed
by considering polynomials in groupoids, instead of sets. For those, the
constructions involved into composition have to be performed up to homotopy,
which is conveniently handled in the setting of homotopy type theory: we use it
here to formally perform the constructions required to build our cartesian
bicategory, in Agda. Notably, this requires us introducing an axiomatization in
a small universe of the type of finite types, as an appropriate higher inductive
type of natural numbers and bijections.

\end{abstract}




\section{Introduction}
Polynomial functors have been introduced as a categorical generalization of
traditional polynomials and have been intensively studied by Kock and
collaborators~\cite{gambino2009polynomial,kock2012data,gepner2017infty}. They
have become an important categorical tool, allowing the definition and
manipulation of various structures such as opetopes~\cite{kock2010polynomial} or
type theories~\cite{arkor2020algebraic}. This motivates the study of the
categorical structures they bear in order to facilitate constructions on those,
and the situation turns out to be quite subtle. For instance, one would expect
that they should be cartesian closed (after all, most usual categories are), but
it is not the case: the cartesian product functor does not have a satisfactory
right adjoint. There are essentially two reasons for that.

The first one is that there are size issues: the naive definition of the
exponential appears to be too large to be a proper object in the category of
polynomials.  This is easily overcome by restricting to polynomials which are
finitary, \ie only involve monomials consisting of products of variables which
are finite. The resulting category is isomorphic to the category of Girard's
normal functors~\cite{girard1988normal}, which is a model of simply typed
$\lambda$\nbd-calculus, and linear logic~\cite{hasegawa2002two}, which can be
thought of as a quantitative variant of the relational model (historically, this
model is in fact one of the starting points motivating the introduction of
linear logic).

The second one is more problematic: polynomial functors carry an intuitive
2-categorical structure, but it was observed early on that the closure mentioned
above fails to extend to a 2\nbd-categorical
one~\cite{girard1988normal,taylor1989quantitative,hasegawa2002two}. Here, we
advocate that a satisfactory answer to this problem is provided by switching
from traditional polynomial functors to ones over groupoids, as first considered
by Kock~\cite{kock2012data}, see
also~\cite{weber2015polynomials,vidmar2018polynomial}. We show that the
resulting bicategory is cartesian closed; it is more generally a model of
intuitionistic linear logic, which we expect to extend as a model of
differential linear logic. The resulting category is close to the ``equivariant
variant'' of polynomials, provided by generalized species or analytic
functors~\cite{fiore2004generalised,fiore2007differential,fiore2008cartesian,fiore2014analytic}.

In order to assist us in our proofs and help us gain confidence in those, we
have formalized most of them, from the beginning, in Agda, in the setting of
univalent type theory~\cite{hottbook} (with a custom implementation of homotopy
type theory). Since this reflects the way we worked, allows easily manipulating
objects such as groupoids as 1-truncated types, and ensures for free the
functoriality and homotopy invariance of all constructions, we decided to
present our results directly in the type-theoretic formalism, and a
``translation'' into the traditional set-theoretic setting for polynomial
functors is planned for future works. Moreover, the present work required us to
develop specific type-theoretic constructions, which we think could find
application outside the scope of this work. A salient contribution is the
construction, as a higher inductive type, of the type of finite sets and
bijections in a small universe. The Agda code is publicly available~\cite{git}.


We recall the traditional definition of polynomials, polynomial functors and
associated constructions in \cref{sec:poly}, we then present a formalization of
the bicategory of polynomial functors in groupoids in \cref{sec:bicategory} and
discuss the cartesian structure and the failure of being properly closed in
\cref{sec:naive}. We then introduce and study the notion of finite type in
\cref{sec:fintype} and finally properly construct a cartesian closed bicategory
in \cref{sec:ccc}.


\section{Polynomial functors}
\label{sec:poly}
We begin by recalling the traditional definition of polynomial functors, as well
as related constructions. All the material in this section is already known, but
required in the following.

\subsection{The category of polynomial functors}
We briefly recall here the categorical generalization of the notion of
polynomial provided by polynomial functors, and refer the reader
to~\cite{gambino2009polynomial,kock2009notes} for a detailed
presentation. Polynomials as traditionally defined as finite sums of the form
$P(X)=\sum_{0\leq i<k}X^{n_i}$. This notion can be ``categorified'' by taking a
set $B$ of monomials (instead of specifying their number $k$) and having a set
$E_b$ of instances of $X$ in each monomial~$b$ (instead of specifying their
number~$n_i$). This data can thus collected as a function $P:E\to B$
representing the \emph{polynomial},
where $B$ is the set of monomials and for each monomial $b$, $E_b=P^{-1}(b)$ is
the set of instances of~$X$ involved in the monomial. Such a function~$P$,
induces a functor~$\intp{P}:\Set\to\Set$ defined, by mimicking the usual
definition of polynomials, as
\[
  \intp{P}(X)=\sum_{b\in B}X^{E_b}
\]
and we call \emph{polynomial functor} such a functor. An interesting point of
view on the above data consists in considering the elements of~$B$ as abstract
\emph{operations}, whose \emph{parameters} are the elements of~$E_b$, so that
$\intp{P}(X)$ corresponds to the set obtained by formally applying the
operations in~$B$ to the required number of elements of the set~$X$.

We will more generally consider the ``typed'' or ``colored'' variant of polynomials and
polynomial functors, where the parameters of an operation are decorated by a
``color'' in a set~$I$, as well as their output in a set~$J$. This data can be
encoded by a diagram~$P$ in~$\Set$ of the form
\begin{equation}
  \label{eq:poly}
  \begin{tikzcd}
    I&\ar[l,"s"']E\ar[r,"p"]&B\ar[r,"t"]&J
  \end{tikzcd}
\end{equation}
which we call a \emph{polynomial} and consists of an uncolored polynomial~$p$
together with functions $s$ and $t$ respectively indicating the colors of the
parameters and outputs of operations.  Note that the previous uncolored setting is
recovered when $I$ and $J$ are both the terminal set. Writing $\Set/I$ for the
slice category of sets over a set~$I$, such data again induces a
functor~$\intp{P}:\Set/I\to\Set/J$, called a \emph{polynomial functor}, obtained
as the composite $\intp{P}=\Sigma_t\circ\Pi_p\circ\Delta_s$
where $\Delta_s$ is the pullback map along~$s$, and $\Sigma_t$ (\resp $\Pi_p$)
is the left (\resp right) adjoint to $\Delta_t$ (\resp $\Delta_p$) given by
post-composition by~$t$ (\resp local cartesian closure).

Alternatively, polynomial functors can be considered as acting on families
instead of slice categories. Given a set~$I$, it is well known that the slice
category $\Set/I$ of sets over~$I$ is equivalent to the category of families
indexed by~$I$,
\begin{equation}
  \label{eq:slice-fam}
  \Set/I\qequivto\Set^I
\end{equation}
and, through this equivalence, the associated polynomial functor is $\intp{P}:\Set^I\to\Set^J$ such that
\[
  \intp{P}\fam{X_i}{i\in I}=\fam{\sum_{b\in t^{-1}(j)}\prod_{e\in p^{-1}(b)}X_{s(e)}}{j\in J}      
\]

It can be shown (and this is non-trivial) that the composite of two polynomial
functors is again polynomial: this means that given two polynomial functors
generated by two diagrams of the form \eqref{eq:poly}, one can find a third
diagram of the form \eqref{eq:poly} which is a presentation for the composite of
the functors. We can thus build a category $\PolyFun$ where an object is a set
and a morphism $I\to J$ is a polynomial functor $\Set/I\to\Set/J$ (or
$\Set^I\to\Set^J$). Note that even though an operation of composition is defined
on polynomials \eqref{eq:poly}, we cannot build a category of those: their
composition being defined by using universal constructions, it will not be
strictly associative and the best we can hope for is a structure of a
bicategory. This motivates investigating the 2-categorical structure of
polynomials.


\subsection{2-categorical structure}
Given two polynomials $P$ and $P'$ of the form~\eqref{eq:poly}, both from~$I$
to~$J$, a morphism between them consists of two functions $\beta:B\to B'$ and
$\epsilon:E\to E'$ between the operations (\resp parameters) of $P$ and those
of~$P'$ making the diagram
\begin{equation}
  \label{eq:pol-mor}
  \begin{tikzcd}
    I\ar[d,equals]&\ar[d,"\varepsilon"']\ar[l,"s"']E\ar[r,"p"]\ar[dr,phantom,very near start,"\lrcorner"]&\ar[d,"\beta"]B\ar[r,"t"]&J\ar[d,equals]\\
    I&\ar[l,"s'"]E'\ar[r,"p'"']&B'\ar[r,"t'"']&J
  \end{tikzcd}  
\end{equation}
commute and such that the middle square is a pullback (such a morphism is
sometimes said to be \emph{cartesian} to insist on this requirement). This last
condition can be understood as requiring that $\beta$ preserves the arity of
operations (it is also technically important because there is no sensible way of
defining horizontal composition of morphisms without this condition). The
composition of polynomials defined above turns out to be associative up to
isomorphism, so that one can define a bicategory $\Poly$ whose 0-cells are sets,
1-cells are polynomials and 2-cells are morphisms of polynomials.

The resulting bicategory can be shown to be biequivalent to the 2-category
$\PolyFun$, obtained by adding suitable natural transformations as 2-cells to the above
category~\cite[Theorem~2.17]{gambino2009polynomial}. There is a subtlety
concerning the 2-cells: the ``natural'' notion of morphism between polynomial
functors, strong natural transformations, is more liberal than the notion of
morphism defined above on corresponding polynomials, and one can either restrict
those transformations (to cartesian natural transformations) or generalize the
notion of morphism between polynomials.

\subsection{Cartesian closed structure}
\label{sec:intro-closure}
The 1-category $\PolyFun$ of polynomial functors is cartesian. Considering
polynomial functors as operating on families through \eqref{eq:slice-fam}, as
explained above, the paring of two polynomial functors~$P:I\to J$ and $Q:I\to K$
is induced by the one in~$\Cat$:
\begin{align*}
  \pair PQ:\Set^I&\to\Set^J\times\Set^K\isoto\Set^{J\sqcup K}
\end{align*}
which motivates defining the product on objects of $\PolyFun$ as the coproduct
of sets.
We can hope that the category also has a closure for products induced by the one
in~$\Cat$. The sequence of bijections of hom-sets
\[
  \Set^I\times\Set^J\to\Set^K
  \qisoto
  \Set^I\to{(\Set^K)}^{\Set^J}
  \qisoto
  \Set^I\to\Set^{\Set^J\times K}
\]
suggests that we define the closure as
\begin{equation}
  \label{eq:naive-closure}
  [J,K]=\Set^J\times K\equivto\Set/J\times K
\end{equation}
However, this does not make sense because the objects of $\PolyFun$ are sets and
this lives in a larger universe. This motivates considering polynomials which
are ``reasonably small''. We say that a polynomial~$P$ of the form
\eqref{eq:poly} is \emph{finitary} when every operation has a finite set of
parameters, \ie the set $E_b=p^{-1}(b)$ is finite for every
operation~$b$. Polynomial functors corresponding to finitary functors can be
shown to be those preserving filtered colimits and are sometimes called
\emph{normal functors}~\cite{girard1988normal,hasegawa2002two}. From now on, we
restrict our category to such polynomials, and consider them up to
isomorphism. This change allows us to replace $\Set/J$ by $\finSet/J$ in the
above definition of the exponential, where $\finSet$ is the class of finite sets (\ie we
consider finite families of set). While this is
certainly ``smaller'', this is still not a proper set; however, it is now
equivalent to a proper set, namely the set $\N/J$ of functions $[n]\to J$ for
some $n\in\N$, where $[n]=\set{0,\ldots,n-1}$ is a canonical choice of a finite
set with $n$ elements. To sum up, it can be shown that the resulting category is
still cartesian and we can define the exponential as $[J,K]=\N/J\times K$.

We have seen that the category of polynomial functors really is a 2-category (or
a bicategory if we consider polynomials instead) and it is natural to expect
that the cartesian closure would extend to the 2-categorical setting, by which
we mean that the isomorphism
\[
  \PolyFun(I\sqcup J,K)
  \isoto
  \PolyFun(I,\N/J\times K)
\]
between sets of isomorphism classes of polynomial functors should extend to an
equivalence of categories. It has however been observed that it is not the case,
see \cite[Remark~2.19]{girard1988normal},
\cite[Example~1.4.2]{taylor1989quantitative} and
\cite[Theorem~1.24]{hasegawa2002two}. As an illustration, consider the
polynomial functor $\intp{P}(X)=X^2$, which is induced by the polynomial~$P$
given by the diagram $1\ot 2\to 1\to 1$
(we write $n$ for a set with $n$ elements). It can be remarked that there are
two automorphisms on the polynomial~$P$: the identity and the morphism of the
form \eqref{eq:pol-mor} where
$\varepsilon:2\to 2$ is the transposition.
The exponential transpose of~$P:0\sqcup 1\to 1$ is the polynomial
$P^\sharp:0\to\N/1\times 1$ (whose target is isomorphic to~$\N$) induced by the
diagram $0\ot 0\to 1\to\N$, where the morphism $1\to\N$ sends the element of~$1$
to $2\in\N$ (which is the arity of the unique operation of~$P$). Because $0$ and
$1$ are respectively initial and terminal in sets, the identity is the only
automorphism of~$P^\sharp$, whereas the two automorphisms of~$P$ should induce
two automorphisms on its exponential transpose.


\subsection{Toward polynomial functors in groupoids}
\label{sec:set-pol}
This tension is solved in \cite{hasegawa2002two} by quotienting the 2-cells
under an ad-hoc equivalence relation. In this paper, we advocate that a more
satisfactory approach consists in switching from polynomial functors in sets to
polynomial functors in groupoids: intuitively, the problem comes from the fact
that, $\N$ being a set, there is no non-trivial endomorphism on a natural
number, which we should have if we were to have a closure for the cartesian
product. By switching to groupoids, we will be able to replace the set $\N$ in
the above construction by the groupoid~$\B$, which also has the natural numbers
as objects, whose morphisms are all automorphisms, such that the group of
automorphism on an object $n$ is the $n$-th symmetric group.

A proper definition of polynomials in groupoids requires more than
simply considering diagrams of the form~\eqref{eq:poly} in the category of
groupoids. For instance, the traditional definition of composition does not
immediately extend to those because the category of groupoids is not locally
cartesian closed, and thus the right adjoint to the change of base functor
$\Delta_f$ is not defined for every morphism~$f$, which prevents us from making
an immediate generalization of the definition of polynomial
functors (a way to address this consists in restricting to
polynomials~\eqref{eq:poly} where $p$ and $t$ are fibrations~\cite{vidmar2018polynomial}).
Following~\cite{kock2012data}, this can be explained as the fact that
switching from sets to groupoids can be thought of as switching from 0-truncated
spaces to 1-truncated spaces, where the strict universal limits and colimits
involved in the constructions on polynomials are not the right ones: we need to
take limits and colimits \emph{up to homotopy}. An important byproduct of
working in homotopy type theory as we do is that all internal constructions are
invariant up to homotopy, which avoids us explicitly dealing with those issues.



\section{The bicategory of polynomial functors}
\label{sec:bicategory}
We now present our formalization in homotopy type theory of the main steps for
constructing the bicategory of polynomial functors. To be precise, we formalize
here the sub-bicategory of the usual one, where we only keep invertible 2-cells
(for which the problem for defining the closure is still non-trivial).
The developments have been performed in Agda, and are available in the
repository~\cite{git} based on our own formalization of homotopy type theory (or
HoTT), following the reference book~\cite{hottbook}, to which we suppose the
reader already acquainted.

We unfortunately do not have enough space here to expose in details the basic
definitions of Agda and homotopy type theory, and only recall some notations. We
write \code{Type} for the universe of small types
and \code{Type₁} for the universe of large types (in particular \code{Type} is
an element of \code{Type₁}). We write $x\equiv y$ for the type of
\emph{identities} (or \emph{equalities} or \emph{paths}) between two terms $x$
and $y$ of the same type. We write $A\simeq B$ for the type of
\emph{equivalences} between two types $A$ and $B$ (possibly in different
universes): it consists of functions $f:A\to B$ admitting an inverse up to
homotopy, in a suitably coherent sense.
We recall that a \emph{proposition} is a type in any two elements are equal,
and a \emph{set} (\resp a \emph{groupoid}) is
a type in which the type $x\equiv y$ of identities between two elements $x$ and
$y$ is a proposition (\resp a set). We postulate here the \emph{univalence}
axiom which states that the canonical map from identities $x\equiv y$ to
equivalences $x\simeq y$ is itself an equivalence.

\subsection{Formalizing polynomials}
Our formalization of polynomials can be found in
\cite[\texttt{Polynomial.agda}]{git}. The direct translation of the
definition~\eqref{eq:poly} of polynomials, as consisting of two types $E$ and
$B$ and three functions $t$, $p$ and $s$ turns out to be quite cumbersome,
because it involves quite a lot of manipulations of identities, even for the
basic constructions of the category of polynomials.
%
%
In the light of the equivalence~\eqref{eq:slice-fam}, it is much more convenient
to take the family point of view instead of the slice one, and use the following
definition, also known in the literature as an \emph{indexed
  container}~\cite{altenkirch2015indexed}. Namely, in a polynomial
\eqref{eq:poly}, the function $t:B\to J$ associates to each operation of the
polynomial in $B$ a color in $J$: this data can equivalently be encoded as a
family of types \code{Op : J → Type} which to every element \code{j} of \code{J}
associates the type \code{Op j} of operations colored by \code{j} of the
polynomial. By performing similar transformations on the rest of the data, we
reach the following definition, which is easier to work with because it uses
more heavily dependent types: the explicit computations we had to perform with
equality above are now implicitly handled by the dependent pattern matching of
Agda.

\begin{definition}
  \label[definition]{def:agda-pol}
  The type of \emph{polynomials} between two types \code{I} and \code{J} is
\begin{verbatim}
record Poly (I J : Type) : Type₁ where
  field
    Op : J → Type
    Pm : (i : I) → {j : J} → Op j → Type
\end{verbatim}  
\end{definition}

\noindent
Given types \code{I} and \code{J}, a polynomial from \code{I} to \code{J}
consists of:
\begin{itemize}
\item a family of types \code{Op j} indexed by the elements \code{j} of
  \code{J}: the \emph{operations} of the polynomial;
\item a family of type \code{Pm i b} indexed by the element \code{i} of
  \code{I} and the operations \code{b} in \code{Pm j} for some \code{j} of
  \code{J} (curly brackets indicate that this argument is usually left implicit): the \emph{parameters} of the polynomial.
\end{itemize}
%

The careful reader will note that the above actually defines ``polynomials in
types'', where there is a type (as opposed to a set or a groupoid) of operations
and parameters. The notion of \emph{polynomial in groupoids} can be obtained by
further restricting to the case where all the involved types (\code{I},
\code{J}, \code{Op j} and \code{Pm i b}) are groupoids, in the sense of HoTT
recalled above.

\subsection{First constructions on polynomials}
The identity polynomial on a type \code{I} is defined as
\begin{verbatim}
Id : Poly I I
Op Id i = ⊤
Pm Id i {j = j} tt = i ≡ j
\end{verbatim}
and has, for each element \code{i} of \code{I}, exactly one operation of
type~\code{i}, whose only parameter is also of type \code{i}. The polynomial
functor induced by a polynomial \code{P} from \code{I} to \code{J} is
\begin{verbatim}
⟦_⟧ : Poly I J → (I → Type) → (J → Type)
⟦_⟧ P X j = Σ (Op P j) (λ c → (i : I) → (p : Pm P i c) → (X i))
\end{verbatim}
Finally, the composite \code{P · Q} of two polynomials \code{P} from \code{I} to
\code{J}, and \code{Q} from \code{J} to \code{K} is
\begin{verbatim}
_·_ : Poly I J → Poly J K → Poly I K
Op (P · Q) = ⟦ Q ⟧ (Op P)
Pm (P · Q) i (c , a) = Σ J (λ j → Σ (Pm Q j c) (λ p → Pm P i (a j p)))
\end{verbatim}
Morphisms between polynomials can be encoded as follows:

\begin{definition}
  \label[definition]{def:agda-pol-mor}
  The type of \emph{morphisms} between two polynomials \code{P} and~\code{Q} is 
\begin{verbatim}
record Poly→ (P Q : Poly I J) : Type where
  field
    Op→ : {j : J} → Op P j → Op Q j
    Pm≃ : {i : I} {j : J} {c : Op P j} → Pm P i c ≃ Pm Q i (Op→ c)
\end{verbatim}  
\end{definition}

\noindent
A morphism thus consists of:
\begin{itemize}
\item a morphism between the operations of \code{P} and those of \code{Q}, which
  respects the typing;
\item a morphism between the parameters of \code{P} and those of~\code{Q}, which
  respects typing and operations, and is moreover an equivalence: this last
  requirement corresponds precisely to imposing that the square in the middle of
  \eqref{eq:pol-mor} is a pullback.
\end{itemize}
We write \code{I ↝ J} (\resp \code{P ↝₂ Q}) for the type of polynomials from
\code{I} to \code{J} (\resp morphisms of polynomials from \code{P} to \code{Q}).

\begin{definition}
  \label[definition]{def:agda-pol-equiv}
  A morphism \code{ϕ} as above is an \emph{equivalence of polynomials} when the
  morphism on operations is an equivalence at each element of \code{J}:
\begin{verbatim}
Poly-equiv : {P Q : I ↝ J} (ϕ : P ↝₂ Q) → Type
Poly-equiv ϕ = {j : J} → is-equiv (Op→ ϕ {j = j})
\end{verbatim}
  We write \code{P ≃₂ Q} for the type of equivalences of polynomials between
  \code{P} and \code{Q}.
\end{definition}

\subsection{A bicategory of polynomials}
Starting from there we can build all the structure one expects to find in a
bicategory of polynomials:
\begin{itemize}
\item we can define the identity polynomial and the composition of
polynomials (see above);
\item we can show that composition of polynomials is associative and unital up
  to an equivalence of polynomials;
\item we can define the horizontal and vertical composition of morphism of
  polynomials;
\item we can show that those compositions are associative and unital up to a
  suitable notion of equivalence of morphisms of polynomials.
\end{itemize}
Moreover, by using univalence, one can show the following.

\begin{proposition}
  \label[proposition]{thm:pol-ua}
  The type of equivalences between two polynomials \code{P} and \code{Q} is
  equivalent to the type $\code{P}\equiv\code{Q}$ of equalities between the two
  polynomials: \code{(P ≡ Q) ≃ (P ≃₂ Q)}.
\end{proposition}

We can therefore build a bicategory (it might be more accurate to call it a
\emph{$(2,1)$-category} since 2-cells are equivalences) of groupoids, polynomials
in groupoids and equivalences, in the following sense.

\begin{definition}
  \label[definition]{def:agda-prebicat}
  A \emph{prebicategory} consists of:
  \begin{itemize}
  \item a type of \code{ob} \emph{objects};
  \item for each objects \code{I} and \code{J}, we have a groupoid
    \code{hom I J} of \emph{morphisms};
  \item for each object \code{I} there is a distinguished morphism \code{id} in
    \code{hom I I} called \emph{identity};
  \item there is a composition operation
\begin{verbatim}
_⊗_ : hom I J → hom J K → hom I K
\end{verbatim}
    for every objects \code{I}, \code{J} and \code{K};
  \item composition is associative and unital:
\begin{verbatim}
assoc  : (P : hom I J) (Q : hom J K) (R : hom K L) →
         (P ⊗ Q) ⊗ R ≡ P ⊗ (Q ⊗ R)
unit-l : (P : hom I J) → id ⊗ P ≡ P
unit-r : (P : hom I J) → P ⊗ id ≡ P
\end{verbatim}
  \item the traditional pentagon law
\begin{verbatim}
ap (λ P → P ⊗ S) (assoc P Q R) ∙ assoc P (Q ⊗ R) S ∙
   ap (λ Q → P ⊗ Q) (assoc Q R S) ≡ assoc (P ⊗ Q) R S ∙ assoc P Q (R ⊗ S)
\end{verbatim}
    and triangle law
\begin{verbatim}
assoc P id Q ∙ ap (λ Q → P ⊗ Q) (unit-l Q) ≡ ap (λ P → P ⊗ Q) (unit-r P)
\end{verbatim}
    of bicategories are satisfied for composable morphisms \code{P}, \code{Q},
    \code{R} and \code{S}.
  \end{itemize}
\end{definition}

\noindent
Above, ``\code{∙}'' denotes the concatenation of paths (or transitivity of
equality) and \code{ap} is a proof that every function is a congruence for
equality.

\begin{theorem}
  \label{thm:pre-bicat}
  There is a prebicategory whose objects are groupoids, 1-cells are polynomials
  in groupoids and 2-cells are equivalences of polynomials.
\end{theorem}

The notion of prebicategory generalizes the notion of precategory in
HoTT~\cite[Section~9.1]{hottbook}. As the name suggests, it lacks a property in
order to bear the name of a bicategory, similarly to the situation with
categories. A morphism \code{P} in \code{hom I J} in a prebicategory is an
\emph{internal equivalence} when there exists morphisms \code{Q} and \code{Q'}
both in \code{hom J I} such that \code{P ⊗ Q ≡ id} and \code{Q' ⊗ P ≡ id}, and
we write \code{I ≃' J} for the type of internal equivalences from \code{I} to
\code{J}. There is a canonical map associating to any equality of type \code{I ≡
  J} between two objects \code{I} and \code{J} an internal equivalence from
\code{I} to \code{J}. A \emph{bicategory} is a prebicategory in which this
canonical map is an equivalence, \ie we have $\code{(I ≡ J) ≃ (I ≃' J)}$ for
every objects \code{I} and \code{J}.

\begin{theorem}
  \label{thm:poly-bicat}
  The prebicategory of \cref{thm:pre-bicat} is a bicategory.
\end{theorem}

\noindent
Note that \cref{thm:pol-ua}, in addition to proving the above theorem, allows us
to use equalities as 2\nbd-cells instead of morphisms in the sense of
\cref{def:agda-pol-equiv}. Because of this, the usual structure of bicategory
which is ``missing'' from \cref{def:agda-prebicat} (\eg horizontal and vertical
composition of 2-cells, the exchange law, etc.) is automatically present thanks
to the general properties of equality. If this was not the case (for instance,
if we wanted to consider morphisms of polynomials instead of equivalences as
2-cells), we would have had to use a much more involved notion of
bicategory~\cite{ahrens2019bicategories}.

\section{Naive cartesian closed structure}
\label{sec:naive}

\subsection{Cartesian structure}
The notion of being cartesian for such a bicategory can be formalized in the
expected way, by requiring the existence of a terminal object \code{T}, a binary
product operation \code{⊕} on objects and projection operations \code{projl :
  hom (I ⊕ J) I} and \code{projr : hom (I ⊕ J) J} for every objects \code{I} and
\code{J}, such that \code{hom I T ≡ ⊤} and the canonical function \code{hom I (J
  ⊕ K) → hom I J × hom I K} (obtained by post-composition with the projection
operations) is an equivalence (and thus an equality by univalence). One can
show:

\begin{theorem}
  The bicategory of \cref{thm:poly-bicat} is cartesian, with the product being
  defined by coproduct \code{⊔} on objects (the groupoids), first
  projection polynomial and pairing operation on polynomials being
  \begin{center}
    \begin{minipage}{0.45\linewidth}
\begin{verbatim}
projl : (I ⊔ J) ↝ I
Op projl i = ⊤
Pm projl (inl i) {i'} tt = i ≡ i'
Pm projl (inr j) {i'} tt = ⊥

\end{verbatim}      
    \end{minipage}
    \begin{minipage}{0.49\linewidth}
\begin{verbatim}
pair : (I ↝ J) → (I ↝ K) → I ↝ (J ⊔ K)
Op (pair P Q) (inl j) = Op P j
Op (pair P Q) (inr k) = Op Q k
Pm (pair P Q) i {inl j} c = Pm P i c
Pm (pair P Q) i {inr k} c = Pm Q i c
\end{verbatim}        
    \end{minipage}
  \end{center}
  (and second projection is similar to first projection).
\end{theorem}

\subsection{Naive closed structure}
\label{sec:agda-naive-closure}
A first step toward constructing a right adjoint to the product is the
formalization of the naive closure described in \cref{sec:intro-closure} given
in~\cite[\code{LargePolynomial.agda}]{git}. Namely, the
formula~\eqref{eq:naive-closure} indicates that the hom space from \code{I} to
\code{J} should be \code{(I → Type) × J}.
Of course, we encounter the same size issues as mentioned in the introduction,
and we need to suppose that \code{Type} is the same as \code{Type₁} (\ie we
disable the checking of universe levels), which makes the logic inconsistent,
for this proof to go through. The formalization is still useful because it is a
simple version of the actual one for the closure, which is more involved but
does not require the extra assumption.

We define the \emph{exponential} of a type as
\begin{verbatim}
Exp : Type → Type₁
Exp I = I → Type
\end{verbatim}
so that the internal hom between two types \code{I} and \code{J} should be
\code{Exp I × J}. We can indeed define a currying map
\begin{verbatim}
curry : (I ⊔ J) ↝ K → I ↝ (Exp J × K)
Op (curry P) (jj , k) = Σ (Op P k) (λ c → ((λ j → Pm P (inr j) c) ≡ jj))
Pm (curry P) i c = Pm P (inl i) (fst c)
\end{verbatim}
which formally transforms an operation with a given set of inputs in
\code{J} into an operation with corresponding family of elements in \code{J} as
output (the inputs in \code{I} are preserved and the output in \code{K} is
preserved as the second component of the output). This could be graphically
pictured as
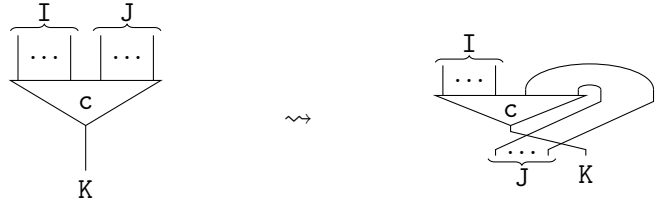
\[
  \begin{tikzpicture}[yscale=.6,baseline=(b.base)]
    \coordinate (b) at (0,0);
    \draw (0,.5) node{\code{c}};
    \draw [decoration={brace}, decorate] (-1,2.1) -- (-.1,2.1) node [pos=0.5,above] {\code{I}};
    \draw [decoration={brace}, decorate] (.1,2.1) -- (1,2.1) node [pos=0.5,above] {\code{J}};
    \draw (-.55,1.5) node {$\ldots$};
    \draw (.55,1.5) node {$\ldots$};
    \draw (-.9,1) -- (-.9,2);
    \draw (-.2,1) -- (-.2,2);
    \draw (.2,1) -- (.2,2);
    \draw (.9,1) -- (.9,2);
    \draw (-1,1) -- (1,1) -- (0,0) -- cycle;
    \draw (0,0) -- (0,-1) node[below]{\code{K}};
  \end{tikzpicture}
  \qquad\qquad\rightsquigarrow\qquad\qquad
  \begin{tikzpicture}[yscale=.4,baseline=(b.base)]
    \coordinate (b) at (0,0);
    \draw (0,.5) node{\code{c}};
    \draw [decoration={brace}, decorate] (-1,2.1) -- (-.1,2.1) node [pos=0.5,above] {\code{I}};
    \draw (-.55,1.5) node {$\ldots$};
    \draw (-.9,1) -- (-.9,2);
    \draw (-.2,1) -- (-.2,2);
    \draw (.2,1) -- (.2,1.2);
    \draw (.9,1) -- (.9,1.2);
    \draw (.9,1.2) arc (180:0:.15);
    \draw (.2,1.2) arc (180:0:.85);
    \draw (1.2,1.2) -- (1.2,.8) -- (-.2,-.8) -- (-.2,-1);
    \draw (1.9,1.2) -- (1.9,.8) -- (.5,-.8) -- (.5,-1);
    \draw (-1,1) -- (1,1) -- (0,0) -- cycle;
    \draw (0,0) -- (0,-.2) -- (1,-.8) -- (1,-1) node[below]{\code{K}};
    \draw (.15,-.9) node {$\ldots$};
    \draw [decoration={brace,mirror}, decorate] (-.3,-1.1) -- (.6,-1.1) node [pos=0.5,below] {\code{J}};
  \end{tikzpicture}
\]
Similarly, one can also define an uncurrying map
\begin{verbatim}
uncurry : I ↝ (Exp J × K) → (I ⊔ J) ↝ K
Op (uncurry P) k = Σ (Exp J) (λ jj → Op P (jj , k))
Pm (uncurry P) (inl i) (jj , c) = Pm P i c
Pm (uncurry P) (inr j) (jj , c) = jj j
\end{verbatim}

\begin{theorem}
  The above maps induce an adjunction: we have
  $\code{((I ⊔ J) ↝ K) ≡ (I ↝ (Exp J × K))}$.
\end{theorem}
\begin{proof}
  We can show that the two above maps are mutually inverse in the sense that we
  have
\begin{verbatim}
(uncurry (curry P) ≃₂ P)
\end{verbatim}
  and
\begin{verbatim}
(curry (uncurry P) ≃₂ P)
\end{verbatim}
  By \cref{thm:pol-ua}, the equivalences of polynomials \code{≃₂} can be turned
  into equalities: \code{curry} and \code{uncurry} thus form an equivalence, and
  we deduce the required equality by univalence.
\end{proof}

An alternative definition for the ``large exponential'' can be given as
follows. The equivalence~\eqref{eq:slice-fam} between slices and families
generalizes in type theory. Given a type \code{I}, we have an equivalence (and
thus an equality by univalence) between types over \code{I} and families of
types indexed by \code{I}, see~\cite[\code{Fam.agda}]{git}:

\begin{theorem}
  \label{thm:agda-slice-fam}
  Given a type \code{I}, we have the equivalence
  \code{(I → Type) ≃ (Σ Type (λ A → A → I))}.
\end{theorem}

\noindent
This indicates that we could equivalently have taken the right member of the
above equivalence as a definition of \code{Exp I} above.


\section{Finite types}
\label{sec:fintype}
Following the plan of \cref{sec:intro-closure} for the proof, we will restrict
to finitary polynomials in order to have a smaller exponential in the closure,
for which we can handle the size issues. In this section, we first formalize the
notion of finiteness for a type, which will then be used to define finitary
functors.

\subsection{Definition and properties}
The proofs associated to this section can be found
in~\cite[\texttt{FinType.agda}]{git}. As customary, we write \code{Fin n} for
the canonical type with $n$ elements, its constructors being the natural numbers
$0$ up to $n-1$. Given a type \code{A}, we write \code{∥ A ∥} for its
\emph{propositional truncation}: an element of this type can be thought of as a
witness that there exists a proof of \code{A}, without explicitly providing such
a proof~\cite[Section~3.7]{hottbook}.

\begin{definition}
  A type is \emph{finite} when it is merely equivalent to the type with $n$
  elements for some natural number~$n$:
\begin{verbatim}
is-finite : Type → Type
is-finite A = Σ ℕ (λ n → ∥ A ≃ Fin n ∥)
\end{verbatim}  
\end{definition}

\begin{remark}
  In case one wonders why we chose to use a truncation in the above definition,
  let us mention that using the definition \code{is-finite' A = Σ ℕ (λ n → A ≃
    Fin n)} would be bad because \code{is-finite' A} is not a proposition,
  \code{is-finite' A} is a large type, and the collection of all finite types is
  actually the natural numbers in the sense that one can construct an
  equivalence \code{Σ Type is-finite' ≃ ℕ}.
\end{remark}


%
The following \cref{lem:cardinality} shows that, for a finite type \code{A},
there is a well-defined notion of \emph{cardinality} which is the natural number
\code{n} such that \code{∥ A ≃ Fin n ∥} holds. In order to show it, we first
need an auxiliary result.

\begin{lemma}
  \label[lemma]{lem:Fin-inj}
  The type constructor \code{Fin} is injective: \code{Fin m ≡ Fin n} implies
  \code{m ≡ n}.
\end{lemma}
\begin{proof}
  Given types \code{A} and \code{B}, a function \code{f : ⊤ ⊔ A → ⊤ ⊔ B} induces
  a function \code{f' : A → B} such that \code{f' a} is \code{f a} if it belongs
  to \code{B}, or \code{f tt} otherwise. Using this construction, one can show
  that \code{⊤ ⊔ A ≃ ⊤ ⊔ B} implies \code{A ≃ B}, and thus that \code{⊤ ⊔ A ≡ ⊤
    ⊔ B} implies \code{A ≡ B} by univalence. We then conclude by induction using
  the fact that \code{Fin (suc n) ≡ ⊤ ⊔ Fin n}. It is also possible to prove
  this fact without resorting to univalence~\cite{kidney2019small}.
\end{proof}

\begin{lemma}
  \label[lemma]{lem:cardinality}
  Given a type \code{A}, if \code{∥ A ≃ Fin m ∥} and \code{∥ A ≃ Fin n ∥} then
  \code{m ≡ n}.
\end{lemma}
\begin{proof}
  The equality between natural numbers is a proposition (\code{ℕ} is a set) and
  we can thus forget about the proposition truncations (by using the
  corresponding elimination rule). By transitivity, we should have \code{Fin m
    ≃ Fin n}, thus \code{Fin m ≡ Fin n} by univalence and thus \code{m ≡ n} by
  injectivity of the type constructor \code{Fin} by \cref{lem:Fin-inj}.
\end{proof}

\noindent
Using this, one can deduce that being finite is a proper predicate:

\begin{proposition}
  \label[proposition]{prop:is-finite-is-prop}
  Being finite for a type is a proposition.
\end{proposition}
\begin{proof}
  Suppose given two proofs of \code{is-finite A} for some type \code{A}. Those
  have the same first component by \cref{lem:cardinality} and the same second
  component by definition of the propositional truncation. They are thus equal.
\end{proof}

\begin{remark}
  As a corollary of previous lemma, we get the fact that \code{is-finite A} is
  equivalent to \code{∥ Σ ℕ (λ n → A ≃ Fin n) ∥}, which could thus have served
  as an alternative definition.
\end{remark}

\noindent
Finite types satisfy the expected basic properties:

\begin{proposition}
  \label[proposition]{prop:fin-closure}
  The types \code{⊥}, \code{⊤} and \code{Fin n} are finite. Finite types are
  closed under (dependent) sums and products.
\end{proposition}

\begin{proposition}
  \label[proposition]{prop:fin-equiv}
  Given equivalent types \code{A} and \code{B}, \code{A} is finite if and only
  if \code{B} is.
\end{proposition}

\begin{proposition}
  \label[proposition]{prop:fin-union}
  Given types \code{A} and \code{B}, we have that \code{A ⊔ B} is finite if and
  only if both \code{A} and \code{B} are finite.
\end{proposition}
\begin{proof}
  Suppose that \code{A ⊔ B} is finite. It can be shown that a decidable subtype
  of \code{Fin n} is of the form \code{Fin k}. More precisely: given a family
  \code{P : Fin n → Type} which is a predicate (i.e. \code{P i} is a proposition
  for every \code{i} in \code{Fin n}) and decidable (i.e. we have \code{(P i) ⊔
    ¬ (P i)}), then there exists a natural number \code{k} such that \code{Σ
    (Fin n) P ≃ Fin k}. By \cref{prop:fin-equiv}, we can deduce that a decidable
  subtype of a finite type is finite. Writing \code{split : A ⊔ B → Bool} for
  the function sending the elements of \code{A} to \code{true} and the elements
  of \code{B} to \code{false}, we have that \code{A} is equivalent to the
  subtype of \code{A ⊔ B} determined by the predicate \code{P : A ⊔ B → Type}
  defined by \code{P x = split x ≡ true}, which is decidable (because \code{Bool} has
  a decidable equality). We conclude that \code{A} is finite, as a decidable
  subtype of \code{A ⊔ B}, which is supposed to be finite (and \code{B} is also
  finite for similar reasons).

  Conversely, if \code{A} and \code{B} are finite, writing \code{m} and \code{n}
  for their cardinal, we have that \code{A ≃ Fin m} and \code{B ≃ Fin n}, and
  thus \code{A ⊔ B ≃ Fin m ⊔ Fin n ≃ Fin (m + n)} (we can ignore the
  propositional truncation because we are eliminating into the proposition of
  being finite).
\end{proof}

\begin{proposition}
  \label[proposition]{prop:finite-set}
  Finite types are sets.
\end{proposition}
\begin{proof}
  Suppose given a finite type \code{(A , (n , F))}, where \code{A} is a type,
  \code{n} is a natural number and \code{F} a proof of \code{∥ A ≃ Fin m
    ∥}. Being a set is a proposition~\cite[Lemma 3.3.5]{hottbook}, by
  elimination of propositional truncation we can thus extract from \code{F} and
  equivalence \code{A ≃ Fin n}, and we conclude by transporting the fact that
  \code{Fin n} is a set (it is easily shown to be decidable and thus a set by
  Hedberg's theorem~\cite[Theorem~7.2.5]{hottbook}).
\end{proof}

We write \code{FinType} for the type of all finite types:
\begin{verbatim}
FinType : Type₁
FinType = Σ Type is-finite
\end{verbatim}
Note that this type is a large one, it lives in \code{Type₁}. By
\cref{prop:is-finite-is-prop}, two elements of this type are equal if and only
if their first components (i.e. underlying types) are equal:

\begin{lemma}
  Given elements \code{A} and \code{B} of \code{FinType}, we have \code{(A ≡ B)
    ≃ (fst A ≡ fst B)}.
\end{lemma}

\noindent
From this, we can deduce:

\begin{proposition}
  \label[proposition]{prop:fintype-gpd}
  \code{FinType} is a groupoid.
\end{proposition}
\begin{proof}
  Given type finite types \code{A} and \code{B}, the type \code{A ≡ B} is
  equivalent to the type of equalities between the underlying types of \code{A}
  and \code{B}, which is a set by \cref{prop:finite-set} as the types of
  equalities between two sets.
\end{proof}

\subsection{A small axiomatization of the type finite types}
In this section, we show an important property: the type \code{FinType} of
finite types is equivalent to a small type~\cite[\texttt{Bij.agda}]{git}. This
is a generalization of the following simple observation: every finite set is
isomorphic to a set of the form $[n]=\set{0,\ldots,n-1}$ for some natural
number~$n$, so that the class of finite sets is equivalent to the set of natural
numbers. However, \code{FinType} is a groupoid and not a set, in the sense of
HoTT: there are non-trivial equalities between finite types which, by
univalence, correspond to isomorphisms of finite types. For this reason, we do
not expect that the type \code{FinType} is equivalent to the traditional type
\code{ℕ} of natural numbers, which has no non-trivial path between its element
(it has decidable equality and thus is a set by Hedberg's theorem), but rather
to a type that we call here \code{𝔹}, which has natural numbers as objects, but
is moreover such that the group of path endomorphisms on an object \code{n} is
the symmetric group on \code{n} elements. A more accurate picture of the
situation than the equivalence between finite sets and natural numbers is thus
the equivalence between the category $\category{Bij}$ of finite sets an
bijections and its skeleton, which has natural numbers as objects. Again, it is
important to remark that $\category{Bij}$ is not small (its collection of
objects does not form a set) whereas its skeleton is.

The type \code{𝔹} being constructed from constructors corresponding to the
natural numbers, but also paths, it is natural to describe it as a higher
inductive type~\cite[Section~6]{hottbook}. Those are not readily available in
current plain version of Agda, but they can be simulated by working
axiomatically with them, as done usually.
The definition we give here is close to the one performed
in~\cite{licata2014eilenberg} in order to define Eilenberg-MacLane spaces in
homotopy type theory.

\begin{definition}
  We axiomatize \code{𝔹} as the small type such that
  \begin{itemize}
  \item it has natural numbers as objects:
\begin{verbatim}
obj : ℕ → 𝔹
\end{verbatim}
  \item for every equivalence between finite sets there is a path in \code{𝔹}
    between the corresponding natural numbers:
\begin{verbatim}
hom : {m n : ℕ} (α : Fin m ≃ Fin n) → obj m ≡ obj n
\end{verbatim}
  \item the path associated to identity is the identity:
\begin{verbatim}
id-coh : (n : ℕ) → hom {n = n} ≃-refl ≡ refl
\end{verbatim}
  \item paths are compatible with composition:
\begin{verbatim}
comp-coh : {m n o : ℕ} (α : Fin m ≃ Fin n) (β : Fin n ≃ Fin o) →
           hom (≃-trans α β) ≡ hom α ∙ hom β
\end{verbatim}
  \item there are no higher paths, \ie \code{𝔹} is a groupoid:
\begin{verbatim}
𝔹-is-groupoid : is-groupoid 𝔹
\end{verbatim}
  \end{itemize}
  We also need to postulate an appropriate elimination principle, which can be
  found in the formalization: it roughly states that, in order to define
  function of type \code{f : 𝔹 → A}, for some groupoid \code{A}, it is enough to
  define:
  \begin{itemize}
  \item an element \code{f (obj n)} of \code{A} for every natural number
    \code{n},
  \item a path \code{apd f (hom e) : f (obj m) ≡ f (obj n)} for every
    equivalence \code{e : Fin m ≃ Fin n},
  \end{itemize}
  in suitably coherent way. The slightly more readable non-dependent version of
  this elimination principle is
\begin{verbatim}
𝔹-rec : (A : Type) → is-groupoid A →
         (obj* : ℕ → A) →
         (hom* : {m n : ℕ} → Fin m ≃ Fin n → obj* m ≡ obj* n) →
         (id-coh* : (n : ℕ) → hom* {n = n} ≃-refl ≡ refl) →
         (comp-coh* : {m n o : ℕ} (α : Fin m ≃ Fin n) (β : Fin n ≃ Fin o) →
                      hom* (≃-trans α β) ≡ hom* α ∙ hom* β) →
         𝔹 → A
\end{verbatim}
  Finally, we also need to postulate two computation rules (implemented as
  rewriting rules). For simplicity we indicate the non-dependent versions
  here. Given a function \code{f : 𝔹 → A} obtained by applying \code{𝔹-rec} to
  some arguments with the above notations, we have
  \begin{itemize}
  \item \code{f (obj n) = obj* n} for any natural number \code{n},
  \item \code{ap f (hom α) = hom* α} for any equivalence \code{α : Fin m ≃ Fin n}.
  \end{itemize}
\end{definition}

\begin{remark}
  In the definition of \code{hom}, the reader might be surprised that we allow
  two different cardinalities of finite sets: given an equivalence \code{Fin m ≃
    Fin n}, we necessarily have \code{m ≡ n}. But this ``more general''
  definition simplifies the proofs in practice.
\end{remark}

There is a canonical function \code{𝔹-to-Fin : 𝔹 → Type} which realizes the
elements of 𝔹 as finite types. More precisely, we define the function
\code{𝔹-to-FinType : 𝔹 → FinType} by using the elimination principle
\code{𝔹-rec} with the following arguments:
\begin{itemize}
\item \code{A} is the type \code{FinType} of finite types, which is a groupoid
  by \cref{prop:fintype-gpd},
\item \code{obj*} states that \code{𝔹-to-FinType (obj n) = Fin n} (through the
  first computation rule),
\item \code{hom*} states that, for an equivalence \code{α : Fin m ≃ Fin n}, we
  have \code{ap 𝔹-to-FinType (hom α)} is the equality \code{Fin m ≡ Fin n}
  obtained from \code{α} by univalence (through the second computation rule).
\end{itemize}
The function \code{𝔹-to-Fin} can then be obtained by post-composing
\code{𝔹-to-FinType} with the first projection, i.e. we forget about the proofs
of finiteness.

One of the main contributions of this paper is the following theorem:

\begin{theorem}
  \label{thm:fin-B}
  The large type of finite types and the above type are equivalent:
  \code{FinType ≃ 𝔹}.
\end{theorem}
\begin{proof}
  The proof uses the ``encode-decode method'' introduced to compute the
  fundamental group of the circle~\cite[Section~8.1]{hottbook}. Given a natural
  number \code{n} and an element \code{b} of~\code{𝔹}, we can encode the type of
  paths of type \code{obj n ≡ b} in \code{𝔹} as the type \code{Code n b} defined
  by induction on \code{b}. More precisely, we define the function \code{Code
    n : 𝔹 → Set} using \code{𝔹-rec} with the following arguments:
  \begin{itemize}
  \item \code{A} is the type \code{Set} of all sets, which is a groupoid,
  \item \code{obj*} states that \code{Code n (obj m)} should be the type
    \code{Fin n ≃ Fin m} of equivalences between \code{Fin m} and \code{Fin n},
  \item \code{hom* α} states, given an equivalence \code{α : Fin m ≃ Fin m'},
    that \code{ap (Code n) α : (Fin n ≃ Fin m) ≡ (Fin n ≃ Fin m')} should be the
    equality between equivalences induced by post-composition with \code{α}.
  \end{itemize}
  Then, we can define an ``encoding'' function \code{e : obj n ≡ b → Code n b}
  such that \code{e refl = ≃-refl}, and a ``decoding'' function \code{d : Code n
    b → obj n ≡ b} by a suitable induction on \code{b}. It can be shown that
  \code{d ∘ e} is the identity and that, for \code{α : Fin m ≃ Fin n}, we have
  \code{e (d (hom α)) ≡ α}.
  This can be used to deduce that the function \code{ap 𝔹-to-FinType : obj m ≡
    obj n → Fin m ≡ Fin n} is an equivalence, a thus that the function
  \code{𝔹-to-FinType} is an embedding (by induction).
  The function \code{𝔹-to-FinType} is easily shown to be surjective: any finite
  type \code{A} is equal to \code{𝔹-to-FinType (obj n)} where \code{n} is the
  cardinality of \code{A}. We finally deduce that the function is an
  equivalence since it is both an embedding and surjective,
  see~\cite[Theorem~4.6.3]{hottbook}.
\end{proof}


\section{A cartesian closed bicategory of finitary polynomials}
\label{sec:ccc}
\subsection{Finitary polynomials}
\label{sec:agda-fin-pol}
Thanks to the notion of finiteness for types, we define
in~\cite[\texttt{FinPolynomial.agda}]{git} finitary polynomials following the
explanations of \cref{sec:intro-closure}.

\begin{definition}
  A polynomial is \emph{finitary} when, for each operation \code{c}, the total
  space of its parameters is finite:
\begin{verbatim}
is-finitary : (P : I ↝ J) → Type
is-finitary P = {j : J} (c : Op P j) → is-finite (Σ I (λ i → Pm P i c))
\end{verbatim}
\end{definition}

\begin{remark}
  Another definition of being finitary could be to require that each space of
  parameters is finite:
\begin{verbatim}
is-finitary' P = (i : I) {j : J} (c : Op P j) → is-finite (Pm P i c)
\end{verbatim}
  but one quickly convinces oneself that this notion is not suitable: identity
  polynomials are not generally finitary in this sense, and being finitary is
  not stable under composition.
\end{remark}

\begin{proposition}
  \label[proposition]{prop:fpol-bicat}
  The cartesian bicategory of \cref{thm:poly-bicat} restricts to a bicategory
  whose morphisms are finitary polynomials.
\end{proposition}
\begin{proof}
  Finitary polynomials can be shown to be stable under the required operations
  (composition, product) using the closure of finite type under the operations
  of~\cref{prop:fin-closure}.
\end{proof}

\noindent
By using similar arguments as in \cref{sec:closure}, the proof of
\cref{sec:agda-naive-closure} can be refined to show that, ignoring size issues,
the cartesian bicategory of finitary polynomials in groupoids admits
\code{Exp I × J} as internal hom, where \code{Exp} is now defined as
\begin{verbatim}
Exp : Type → Type₁
Exp I = Σ (I → Type) (λ F → is-finite (Σ I F))
\end{verbatim}
\ie we restrict to families whose total space is finite. Moreover, through the
equivalence between families and slices (see \cref{thm:agda-slice-fam}), it can
be shown that the above type is equivalent to the one of finite sets
over~\code{I}:

\begin{lemma}
  \label[lemma]{lem:exp-fintype}
  The above definition of the exponential is equivalent to the following one:
\begin{verbatim}
Exp : Type → Type₁
Exp I = Σ FinType (λ N → fst N → I)
\end{verbatim}  
\end{lemma}


\noindent
In turn, by \cref{thm:fin-B}, we have that this type is equivalent to the one of
elements of \code{𝔹} over \code{I}:

\begin{lemma}
  The above definition of the exponential is equivalent to the following one:
\begin{verbatim}
Exp : Type → Type
Exp I = Σ 𝔹 (λ b → 𝔹-to-Fin b → I)
\end{verbatim}
\end{lemma}

\noindent
This is now small type, and thus a reasonable candidate for the right
definition of the exponential.


\subsection{A cartesian closed bicategory}
\label{sec:closure}
We can finally show that the cartesian bicategory of finitary polynomials
(\cref{prop:fpol-bicat}) is closed. Before doing so, we first prove a useful
lemma:

\begin{lemma}
  \label[lemma]{lem:fin-fam-proj}
  Given types \code{I} and \code{J} and a family \code{A : (I ⊔ J) → Type}, we
  have that the total type \code{Σ (I ⊔ J) A} is finite if and only if both
  types \code{Σ I (A ∘ inl)} and \code{Σ J (A ∘ inr)} are finite.
\end{lemma}
\begin{proof}
  We have that \code{Σ (I ⊔ J) A} is equivalent to \code{(Σ I (A ∘ inl)) ⊔ (Σ J
    (A ∘ inr))} and we conclude by \cref{prop:fin-equiv} and
  \cref{prop:fin-union}.
\end{proof}

As suggested in previous section, we define the exponential as
\begin{verbatim}
Exp : Type → Type
Exp I = Σ 𝔹 (λ b → 𝔹-to-Fin b → I)
\end{verbatim}
and finally show our main theorem, by adapting the naive constructions of
\cref{sec:agda-naive-closure} \cite[\texttt{FinPolynomial.agda},
\texttt{CurryUncurry.agda}, \texttt{UncurryCurry.agda}]{git}:

\begin{theorem}
  The cartesian bicategory of finitary polynomials is closed, with \code{Exp I ×
    J} as internal hom from \code{I} to \code{J}.
\end{theorem}
\begin{proof}
  We can define a function
\begin{verbatim}
curry : ((I ⊔ J) ↝ K) → (I ↝ (Exp J × K))
\end{verbatim}
  essentially as in \cref{sec:agda-naive-closure} except that we have to turn
  the output \code{jj : J → Type} of the currying of an operation \code{c}
  into an element of \code{Exp J}. By \cref{lem:fin-fam-proj}, \code{jj} is a
  finite type, and thus induces an element of \code{𝔹} by applying the canonical
  function \code{card : FinType → 𝔹} given \cref{thm:fin-B} which can serve as
  first component in \code{Exp J}, and the second component can be deduced from
  the proof that \code{card} is an inverse to \code{𝔹-to-FinType}. The fact that
  the resulting polynomial is finite follows from
  \cref{lem:fin-fam-proj}. Similarly, we can define a function
\begin{verbatim}
uncurry : (I ↝ (Exp J × K)) → ((I ⊔ J) ↝ K)
\end{verbatim}
  which is again given as in \cref{sec:agda-naive-closure}, using
  \cref{thm:fin-B} to use elements of the exponential and
  \cref{lem:fin-fam-proj} to prove the required finiteness of types. The two
  functions can be shown to be mutually inverse. Finally, the bijection should
  be checked to be natural (this last part is not fully formalized yet).
\end{proof}


\section{Future work}
In this work, we have constructed a cartesian closed bicategory of finitary
polynomial functors in groupoids, in the setting of type theory. A natural next
step is to translate these constructions in to the traditional set-theoretic
setting, in order to ease the comparison with more traditional approaches to
polynomial functors. As mentioned in \cref{sec:set-pol}, we do not expect this
task to be particularly easy.

A second aspect in which we wish to push investigations on this bicategory is its
relation with linear logic. Namely, the bicategory of spans in groupoids can be
understood as the full subbicategory of polynomials of the form~\eqref{eq:poly}
where the morphism~$p$ is the identity. This subbicategory is monoidal with the
tensor induced on objects by cartesian product of groupoids. The inclusion
functor admits a left adjoint which is a strong lax monoidal functor. This
allows to model the exponential of linear logic~\cite{mellies2009categorical}
and will be detailed elsewhere. We also expect that the constructions of
differential linear logic can be interpreted in the model.

Finally, our model is close to the one of generalized species, which is also a
model of (differential) linear
logic~\cite{fiore2004generalised,fiore2008cartesian} inspired by combinatorial
species \cite{joyal1986foncteurs}. We would like to understand the relationship
between those two models, which is hinted at
in~\cite[Section~3.9]{kock2012data}: one of the main difference is that whereas
generalized species are be composed using traditional composition of
profunctors, which involves a quotient, homotopy polynomial functors perform a
homotopy quotient, and we should be able to obtain the first from the second by
suitably discarding homotopical information (\ie ``taking $\pi_0$'').

On the long term, we finally want to investigate the generalization to
polynomials in $\infty$-groupoids, as first studied
in~\cite{gepner2017infty}. This would make more transparent the comparison with
the type theoretic formalization, but we expect the study of this model to be
much more involved on a technical level. The definitions given here should work
identically without the hypothesis that types are groupoids, and the resulting
polynomials should organize as an $\infty$-category instead of a bicategory, but
there is currently no known definition of $\infty$-categories inside type
theory, preventing us from formally showing this result for now.

\bibliographystyle{eptcs}
\bibliography{papers}
\end{document}